\def \VersionForArXiV {}
\ifdefined\VersionForArXiV%
	\newcommand{\LongVersion}[1]{#1}
	\newcommand{\ACMVersion}[1]{}
\else
	\newcommand{\LongVersion}[1]{}
	\newcommand{\ACMVersion}[1]{#1}
\fi

\ifdefined\VersionForArXiV
	\documentclass[twocolumn,10pt]{article}
	\usepackage[svgnames,table]{xcolor}
	\usepackage{amsmath} %
	\usepackage{amssymb} %
\else
	\PassOptionsToPackage{svgnames,table}{xcolor}
	\documentclass[sigplan,screen]{acmart}
\fi

\usepackage[utf8]{inputenc}
\usepackage[english]{babel}
\usepackage{csquotes}

\usepackage[ruled,vlined,linesnumbered]{algorithm2e}
	\SetKwInOut{Input}{input}
	\SetKwInOut{Output}{output}

\usepackage{subcaption}

\usepackage{paralist} %
\usepackage{adjustbox}

\newenvironment{ienumeration}
	{\ifdefined\VersionForArXiV\begin{enumerate}\else\begin{inparaenum}[\itshape i\upshape)]\fi}
	{\ifdefined\VersionForArXiV\end{enumerate}\else\end{inparaenum}\fi}

\newenvironment{oneenumeration}
	{\ifdefined\VersionForArXiV\begin{enumerate}\else\begin{inparaenum}[1)]\fi}
	{\ifdefined\VersionForArXiV\end{enumerate}\else\end{inparaenum}\fi}

\ifdefined\VersionForArXiV
	\usepackage[backend=biber,backref=true,style=alphabetic,url=false,doi=true,defernumbers=true,sorting=anyt,maxnames=99]{biblatex} %
	\DeclareSourcemap{
		\maps[datatype=bibtex]{
			\map[overwrite]{
				\step[fieldsource=toappear, match={true},
					fieldset=note, fieldvalue={To appear}, append]
			}
		}
	}

	\DeclareSourcemap{
	\maps[datatype=bibtex]{
		\map{
			\step[fieldsource=NOdoi, fieldtarget=doi]
			\step[fieldsource=NOeditor, fieldtarget=editor]
			\step[fieldsource=NOlocation, fieldtarget=location]
			\step[fieldsource=NOmonth, fieldtarget=month]
		}
	}
	}

	\renewbibmacro*{doi+eprint+url}{%
		\iftoggle{bbx:doi}
			{\color{black!40}\footnotesize\printfield{doi}}
			{}%
		\newunit\newblock
		\iftoggle{bbx:eprint}
			{\usebibmacro{eprint}}
			{}%
		\newunit\newblock
		\iftoggle{bbx:url}
			{\usebibmacro{url+urldate}}
			{}%
	}

\fi
\ifdefined\VersionWithComments
	\usepackage{draftwatermark}
	\SetWatermarkText{draft}
	\SetWatermarkScale{2}
	\SetWatermarkColor[gray]{0.9}
\fi
\definecolor{darkblue}{rgb}{0, 0, 0.7}

\ifdefined\VersionForArXiV
\usepackage[
		pdfauthor={Etienne Andre, Shapagat Bolat, Engel Lefaucheux, Dylan Marinho},%
		pdftitle={strategFTO: Untimed control for timed opacity},
		breaklinks  = true,
		colorlinks  = true,
		citecolor   = blue!50!blue,
		linkcolor   = darkblue,
		urlcolor    = blue!50!green,
	]{hyperref}
\fi

\usepackage[capitalise,english,nameinlink]{cleveref} %
\crefname{line}{\text{line}}{\text{lines}} %

\usepackage{listings}

\definecolor{mygreen}{rgb}{0,0.6,0}
\definecolor{mygray}{rgb}{0.5,0.5,0.5}
\definecolor{mymauve}{rgb}{0.58,0,0.82}

\newcommand{\defProblem}[3]
{%
	\noindent\fcolorbox{black}{blue!15}{
	\begin{minipage}{.95\columnwidth}
		\textbf{#1 problem:}\\
		\textsc{Input}: #2\\
		\textsc{Problem}: #3
	\end{minipage}
}

	\smallskip

}

\usepackage{tikz}
\usepackage{pgfplots}
\pgfplotsset{compat=1.15}
\usepackage{mathrsfs}
\usetikzlibrary{arrows}
\usetikzlibrary{patterns}
\usetikzlibrary{calc}
\usetikzlibrary{arrows,automata,positioning,shapes} %
\tikzstyle{pta}=[auto, ->, >=stealth']
\tikzstyle{every node}=[initial text=]
\tikzstyle{location}=[rectangle, rounded corners, minimum size=12pt, draw=black, fill=blue!10, inner sep=2pt]
\tikzstyle{invariant}=[draw=black, dotted, inner sep=1pt, node distance=0] %
\tikzstyle{final}=[double, fill=blue!50]

\tikzstyle{urgent}=[fill=yellow, thick, dotted] %
\tikzstyle{private}=[fill=red,thick]

\usepackage{pgfplots}
\usepackage{siunitx}
\usepackage{csvsimple}
\pgfplotsset{
	legend entry/.initial=,
	every axis plot post/.code={%
		\pgfkeysgetvalue{/pgfplots/legend entry}\tempValue
		\ifx\tempValue\empty
		\pgfkeysalso{/pgfplots/forget plot}%
		\else
		\expandafter\addlegendentry\expandafter{\tempValue}%
		\fi
	},
}

\usetikzlibrary{backgrounds,fit}
\tikzset{
	background/.style={
		draw,
		fill=gray!30,
		rounded corners,
		inner sep=.3cm
	},
	backgroundForAll/.style={
		draw,
		fill=purple!30,
		rounded corners,
		inner sep=.1cm,
		prefix after command= {\pgfextra{\tikzset{every label/.style={fill=purple!30,draw=purple}}}}
	}
}

\definecolor{coloract}{rgb}{0.50, 0.70, 0.30}
\definecolor{colorclock}{rgb}{0.4, 0.4, 1}
\definecolor{colordisc}{rgb}{1, 0, 1}
\definecolor{colorloc}{rgb}{0.4, 0.4, 0.65}
\definecolor{colorparam}{rgb}{1, 0.6, 0.0}

\definecolor{loccolor1}{rgb}{1, 0.3, 0.3}
\definecolor{loccolor2}{rgb}{0.3, 1, 0.3}
\definecolor{loccolor3}{rgb}{0.3, 0.3, 1}
\definecolor{loccolor4}{rgb}{1, 0.3, 1}
\definecolor{loccolor5}{rgb}{1, 1, 0.3}
\definecolor{loccolor6}{rgb}{0.3, 1, 1}
\definecolor{loccolor7}{rgb}{0.9, 0.6, 0.2}
\definecolor{loccolor8}{rgb}{0.7, 0.4, 1}
\definecolor{loccolor9}{rgb}{0.5, 1, 0.75}
\definecolor{loccolor10}{rgb}{0.8, 0.7, 0.6}
\definecolor{loccolor11}{rgb}{0.6, 0.7, 0.8}
\definecolor{loccolor12}{rgb}{0.2, 0.5, 0.9}
\definecolor{loccolor13}{rgb}{0.5, 0.9, 0.2}
\definecolor{loccolor14}{rgb}{0.9, 0.2, 0.5}
\definecolor{loccolor15}{rgb}{0.7, 0.7, 0.7}
\definecolor{loccolor16}{rgb}{0.8, 0.8, 0.5}

\newcommand{\styleact}[1]{\ensuremath{\textcolor{coloract}{{\mathit{#1}}}}}
\newcommand{\styleclock}[1]{\ensuremath{\textcolor{colorclock}{{#1}}}}
\newcommand{\styledisc}[1]{\ensuremath{\textcolor{colordisc}{\mathrm{#1}}}}

\newcommand{\clockx}{\ensuremath{\styleclock{x}}}
\newcommand{\clocky}{\ensuremath{\styleclock{y}}}

\newcommand{\stylecode}[1]{\textcolor{colorloc}{\texttt{#1}}}

\newcommand{\ATMaskPassword}{\ensuremath{\styleact{askPwd}}}
\newcommand{\ATMcorrectAmount}{\ensuremath{\styleact{correctAmount}}}
\newcommand{\ATMcorrectPassword}{\ensuremath{\styleact{correctPwd}}}
\newcommand{\ATMrequestBalance}{\ensuremath{\styleact{reqBalance}}}
\newcommand{\ATMfinish}{\ensuremath{\styleact{finish}}}
\newcommand{\ATMincorrectAmount}{\ensuremath{\styleact{incorrectAmount}}}
\newcommand{\ATMincorrectPassword}{\ensuremath{\styleact{incorrectPwd}}}
\newcommand{\ATMnormalWithdrawal}{\ensuremath{\styleact{normalWithdraw}}}
\newcommand{\ATMpressFinish}{\ensuremath{\styleact{pressFinish}}}
\newcommand{\ATMpressOK}{\ensuremath{\styleact{pressOK}}}
\newcommand{\ATMquickWithdrawal}{\ensuremath{\styleact{quickW\LongVersion{ith}draw}}}
\newcommand{\ATMrestart}{\ensuremath{\styleact{restart}}}
\newcommand{\ATMstart}{\ensuremath{\styleact{start}}}
\newcommand{\ATMtakeCash}{\ensuremath{\styleact{\textcolor{red}{takeCash}}}}

\newcommand{\rowHeader}{\rowcolor{blue!20}\bfseries}
\newcommand{\cellYes}{\cellcolor{green!20}\textbf{$\surd$}}

\newcommand{\init}{_0}

\newcommand{\A}{\ensuremath{\mathcal{A}}}
\newcommand{\TA}{\A}

\newcommand{\Actions}{\ensuremath{\Sigma}}
\newcommand{\ControllableActions}{\ensuremath{\Sigma_c}}
\newcommand{\UncontrollableActions}{\ensuremath{\Sigma_u}}
\newcommand{\action}{\ensuremath{a}}

\newcommand{\assign}{\leftarrow}

\newcommand{\Clock}{\mathbb{X}} %
\newcommand{\ClockCard}{H} %
\newcommand{\clock}{x} %
\newcommand{\clockval}{\mu} %
\newcommand{\ClocksZero}{\vec{0}}
\newcommand{\compOp}{\bowtie}

\newcommand{\Control}{\ensuremath{\mathit{Control}}}

\newcommand{\duration}{\ensuremath{\mathit{dur}}}
\newcommand{\edge}{e}
\newcommand{\Edges}{E}

\newcommand{\longuefleche}[1]{\stackrel{#1}{\longrightarrow}}
\newcommand{\longueflecheRel}[1]{\stackrel{#1}{\mapsto}}

\newcommand{\flecheRel}{{\rightarrow}}

\newcommand{\guard}{g}

\newcommand{\invariant}{I}
\newcommand{\K}{K}

\newcommand{\loc}{\ensuremath{\ell}} %
\newcommand{\locinit}{\loc\init}
\newcommand{\Loc}{L} %
\newcommand{\locfinal}{\ensuremath{\loc_f}}
\newcommand{\locpriv}{\ensuremath{\loc_{\mathit{priv}}}}

\newcommand{\R}{{\mathbb{R}}}

\newcommand{\sinit}{s\init} %
\newcommand{\strategy}{\ensuremath{\sigma}}
\newcommand{\strategies}{\ensuremath{\mathcal{S}}}
\newcommand{\concstate}{\ensuremath{s}} %
\newcommand{\States}{S} %

\newcommand{\TTS}{\ensuremath{T}}
\newcommand{\varrun}{\rho} %

\newcommand{\setN}{\ensuremath{\mathbb N}}

\newcommand{\setR}{\ensuremath{\mathbb R}}
\newcommand{\setRgeqzero}{\ensuremath{\setR_{\geq 0}}}

\newcommand{\setZ}{\ensuremath{\mathbb Z}}

\newcommand{\PrivDurReach}[1]{\ensuremath{\mathit{DReach}^\mathit{priv}(#1)}}
\newcommand{\PubDurReach}[1]{\ensuremath{\mathit{DReach}^{\neg \mathit{priv}}(#1)}}

\newcommand{\PrivReach}[3]{\ensuremath{\mathit{Reach}^{#1}_{#2}(#3)}}
\newcommand{\PubReach}[3]{\ensuremath{\mathit{Reach}^{#1}_{\neg #2}(#3)}}

\newcommand{\set}[1]{\ensuremath{\left\{ #1 \right\}}}

\newcommand{\resets}{R}

\newcommand{\reset}[2]{\ensuremath{[#1]_{#2}}}

\newcommand{\stylealgo}[1]{\ensuremath{\textsf{#1}}}
\newcommand{\synthControl}{\stylealgo{synthCtrl}}
\newcommand{\synthMinControl}{\stylealgo{synthMinCtrl}}
\newcommand{\synthMaxControl}{\stylealgo{synthMaxCtrl}}
\newcommand{\witnessMinControl}{\stylealgo{witnessMinCtrl}}
\newcommand{\witnessMaxControl}{\stylealgo{witnessMaxCtrl}}

\newcommand{\ComplexityFont}[1]{{\sffamily\upshape #1}}

\newcommand{\TEXPTIME}{\ComplexityFont{3EXPTIME}\xspace}

\newcommand{\PSPACE}{\ComplexityFont{PSPACE}\xspace}

\ifdefined\VersionForArXiV
	\usepackage{amsthm}
	\theoremstyle{plain}

	\newtheorem{proposition}{Proposition}

	\theoremstyle{definition}
	\newtheorem{definition}{Definition}
	\newtheorem{example}{Example}

	\theoremstyle{remark}

\fi

\newcommand{\strategFTO}{\textsf{strategFTO}}
\newcommand{\imitator}{\textsf{IMITATOR}}
\newcommand{\PolyOp}{\textsc{PolyOp}}

\ifdefined \VersionWithComments
 	\definecolor{colorok}{RGB}{80,80,150}
\else
	\definecolor{colorok}{RGB}{0,0,0}
\fi

\newcommand{\eg}{\textcolor{colorok}{e.\,g.,}\xspace}

\newcommand{\ie}{\textcolor{colorok}{i.\,e.,}\xspace}
\newcommand{\st}{\textcolor{colorok}{s.t.}\xspace}

\newcommand{\wrt}{{w.r.t.}\xspace} %

\newcommand{\orcidID}[1]{\href{https://orcid.org/#1}{\includegraphics[width=1em]{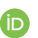}}}

\newcommand{\ourtitle}{\strategFTO{}: Untimed control for timed opacity}

\newcommand{\ouracks}{%
	This work is partially supported by the ANR-NRF French-Singaporean research program \href{https://www.loria.science/ProMiS/}{ProMiS} (ANR-19-CE25-0015 / 2019 ANR NRF 0092)
	and
	the ANR research program BisoUS. %
	
	Experiments presented in this paper were carried out using the Grid'5000 testbed, supported by a scientific interest group hosted by Inria and including CNRS, RENATER and several universities as well as other organizations (see \url{https://www.grid5000.fr}).
}

\newcommand{\ourkeywords}{opacity, timing leak, timed automata, security, control, \imitator{}}

\newcommand{\ourabstract}{%
	We introduce a prototype tool \strategFTO{} addressing the verification of a security property in critical software.
	We consider a recent definition of timed opacity where an attacker aims to deduce some secret while having access only to the total execution time.
	The system, here modelled by timed automata, is deemed opaque if
	for any execution time, there are either no corresponding runs, or both public and
	private corresponding runs.
	We focus on the untimed control problem:
	exhibiting a controller, \ie{} a set of allowed actions, such that the system restricted to those actions is fully timed-opaque.
	We first show that this problem is not more complex than the full timed opacity problem, and then we propose an algorithm, implemented and evaluated in practice.
}

\ifdefined\VersionForArXiV
\fi

\begin{document}

\ifdefined\VersionForArXiV\else

	\title{\ourtitle{}}

	\author{Étienne André}
	\orcid{0000-0001-8473-9555}
	\affiliation{%
	\institution{Université Sorbonne Paris Nord, LIPN, CNRS UMR 7030%
}
	\city{Villetaneuse}
	\country{France}
	\postcode{93430}
	}

	\author{Shapagat Bolat}
	\orcid{0000-0002-6035-5319}
	\author{Engel Lefaucheux}
	\orcid{0000-0003-0875-300X}

	\author{Dylan Marinho}
	\orcid{0000-0002-2548-6196}

	\affiliation{%
	\institution{Université de Lorraine, CNRS, Inria, LORIA}
	\city{Nancy}
	\country{France}
	\postcode{54\,000}
	}

\begin{abstract}
	\ourabstract{}
\end{abstract}


\maketitle

\fi
\ifdefined\VersionForArXiV

	\title{\textbf{\ourtitle{}}%
		\footnote{%
			This manuscript is the author (and slightly extended) version of the manuscript of the same name published in the proceedings of the 8th International Workshop on Formal Techniques for Safety-Critical Systems (\href{https://2022.splashcon.org/home/ftscs-2022}{FTSCS 2022}).
			The final authenticated version is available at \href{https://dl.acm.org/}{\nolinkurl{dl.acm.org}}.
			\ouracks{}
		}
	}

	\author{}
	\date{}

	\makeatletter
	\g@addto@macro\@maketitle{
	\begin{center}
	Étienne André$^{1,2}$\orcidID{0000-0001-8473-9555}, Shapagat Bolat$^{2}$\orcidID{0000-0002-6035-5319}, Engel Lefaucheux$^{2}$\orcidID{0000-0003-0875-300X}, and Dylan Marinho$^{2}$\orcidID{0000-0002-2548-6196}
	\end{center}

	\bigskip

	\noindent{\footnotesize
		$^1$LIPN, CNRS UMR 7030, Université Sorbonne Paris Nord\\
		$^2$Université de Lorraine, CNRS, Inria, LORIA, Nancy, France\\
	}

	\bigskip

	\newcommand{\keywords}[1]
	{%
		\small\textbf{\textit{Keywords---}} #1
	}

	\begin{abstract}
		\ourabstract{}
	\end{abstract}


	\bigskip
	\bigskip
	\bigskip
	}
	\makeatother

	\thispagestyle{plain}

	\maketitle

\fi

\ifdefined \VersionWithComments
	\textcolor{red}{\textbf{This is the version with comments. To disable comments, comment out line~3 in the \LaTeX{} source.}}
\fi

\section{Introduction}\label{section:introduction}
We address here the control of timed systems to avoid timing leaks, \ie{} the leakage of private information that can be deduced from time.
We use as underlying model timed automata (TAs)~\cite{AD94}, an extension of finite-state automata with real-valued clocks.%
\LongVersion{%

\paragraph{Context}}
Opacity is a key security property requiring that
an external user should not be able to deduce whether the execution of a system
contains a secret behavior through its observation.
This property was first formalized for labeled transition systems~\cite{BKMR08}, by specifying a subset of secret paths and requiring that, for any secret path, there is a non-secret one with the same observation. 
Opacity raises challenging research issues such as
\begin{oneenumeration}%
	\item specifying formally opacity in various frameworks~\cite{HughesS04,BKMR08}, 
	\item verifying opacity properties~\cite{Mazare05,BKMR08}, and
	\item developing mechanisms to design a system satisfying opacity while preserving functionality and performance~\cite{BGIRSVY12,BHL18}.
\end{oneenumeration}

Franck Cassez proposed in~\cite{Cassez09} a first definition of \emph{timed opacity} asking whether an attacker can deduce a secret by observing a set of observable actions together with their timestamp.
He proved that opacity is undecidable for TAs%
, mainly from the undecidability of the language inclusion problem for TAs~\cite{AD94}%
.
\LongVersion{The opacity problem is also undecidable for the restricted subclass of event-recording automata~\cite{AFH99}.
}%
Based on this definition of opacity, some decidable subclasses were proposed, for real-time automata~\cite{WZ17,WZA18} (a severely restricted subclass of TAs with a single clock), or over bounded-time~\cite{AEYM21}.

In~\cite{ALMS23}, we proposed a definition of opacity where the attacker only has access (in addition to the model knowledge) to the system \emph{execution time}, \ie{} the time from the initial location to a given location.
The timed opacity problem therefore asks ``for which execution times is the attacker unable to deduce whether a private location was visited?''
The \emph{full} timed opacity problem asks whether the system is timed-opaque for all execution times, \ie{} the attacker is never able to deduce whether the private location was visited by an execution.
We proved in~\cite{ALMS23} that this latter problem is decidable (in \TEXPTIME), and we proposed a practical algorithm using a parametric version of TAs~\cite{AHV93}, implemented in \imitator{}~\cite{Andre21}.

\paragraph{Contribution}
If a system is not fully timed-opaque, there may be ways to tune it to enforce opacity.
For instance, one could change internal delays, or add some \stylecode{sleep()} or \stylecode{Wait()} statements in the program (see \eg{} \cite{ALMS23}).
In this paper, we consider a static (untimed) form of control of the system.
This indicates whether there is a way of restricting the behavior of users to ensure full timed opacity.
With that mindset, we assume the set of actions of the TA is partitioned into a set of \emph{controllable} actions (that can be disabled) and a set of \emph{uncontrollable} actions (that cannot be disabled).
We address the following goal: exhibit a controller (\ie{} a subset of the system controllable actions to be kept in addition to the uncontrollable actions, while other controllable actions are disabled) guaranteeing the system to be fully timed-opaque.
We propose an algorithm exhibiting a set of controllers ensuring opacity, implemented into a tool \href{https://github.com/DylanMarinho/Controlling-TA}{\strategFTO{}}, calling \imitator{}~\cite{Andre21} for computing suitable opaque execution times, and \PolyOp{}~\cite{BHZ08} for additional polyhedra operations.

\paragraph{Related works}
It is well known that observing the time taken by a system to finish some operation is a
potential way to get information out of it (see \eg{} \cite{Kocher96}). As such,
identifying which information is released by the timing of a system has been studied
both from a security and a safety perspective.

From the security point of view, beyond the works related to timed opacity and TAs~\cite{Cassez09,WZ17,WZA18,AEYM21,ALMS23}, the notion of non-interference has been
widely studied. 
	A first definition of timed non-interference was proposed for TAs in~\cite{BDST02,BT03}. This notion is extended to PTAs in~\cite{AK20}, with a semi-algorithm implemented using \imitator{}~\cite{Andre21}.
In~\cite{GMR07}, another notion of timed interference called
timed strong non-deterministic non-interference (SNNI) which was based on timed language equivalence between the automaton with hidden low-level actions and the automaton with removed low-level actions was developed.
This notion is in some aspects stronger than the opacity notion we consider, and is undecidable.
SNNI was adapted in~\cite{VNN18} to allow some intentional information leakage and
a form of control aimed at ensuring it was presented in~\cite{BCLR15}.
Their framework gives to the attacker more information than the total execution time, and their control differs from ours to include that knowledge.

The \emph{diagnosis} of TAs is one of the dominant research directions aimed at analysing 
information leakage from a safety perspective. Its goal is to detect, by observing the 
system, whether some faulty behavior occurred. As such, it is some form of dual to
opacity.
Diagnosis was first introduced for TAs in~\cite{Tripakis02}. Diagnosability of a system
is shown there to be decidable, though the actual diagnoser may be quite complex 
(see~\cite{BCD05} for subclasses of TAs allowing simpler diagnoser, see also~\cite{CT09} for a summary of the main results on the diagnosis of TAs and~\cite{Cassez10} for a
diagnosability focused control of TAs).

\section{Preliminaries}\label{section:preliminaries}

We assume a set~$\Clock = \{ \clock_1, \dots, \clock_\ClockCard \} $ of \emph{clocks}, \ie{} real-valued variables that all evolve over time at the same rate.
A clock valuation is a function
$\clockval : \Clock \rightarrow \setRgeqzero$.
We write $\ClocksZero$ for the clock valuation assigning $0$ to all clocks.
Given $d \in \setRgeqzero$, $\clockval + d$ denotes the valuation \st{} $(\clockval + d)(\clock) = \clockval(\clock) + d$, for all $\clock \in \Clock$.
Given $\resets \subseteq \Clock$, we define the \emph{reset} of a valuation~$\clockval$, denoted by $\reset{\clockval}{\resets}$, as follows: $\reset{\clockval}{\resets}(\clock) = 0$ if $\clock \in \resets$, and $\reset{\clockval}{\resets}(\clock)=\clockval(\clock)$ otherwise.

A clock guard~$\guard$ is a constraint over $\Clock$ defined by a conjunction of inequalities of the form
$\clock \compOp d$\label{def:clockguards}, with
	$ d \in \setZ$ and ${\compOp} \in \{<, \leq, =, \geq, >\}$.
Given~$\guard$, we write~$\clockval\models\guard$ if %
the expression obtained by replacing each~$\clock$ with~$\clockval(\clock)$ in~$\guard$ evaluates to true.

\begin{definition}[TA~\cite{AD94}]\label{def:TA}
	A TA $\TA$ is a tuple \mbox{$\TA = (\Actions, \Loc, \locinit, \locpriv, \locfinal, \Clock, \invariant, \Edges)$}, where:
	\begin{ienumeration}
		\item $\Actions$ is a finite set of actions,
		\item $\Loc$ is a finite set of locations,
		\item $\locinit \in \Loc$ is the initial location,
		\item $\locpriv \in \Loc$ is the private location,
		\item $\locfinal \in \Loc$ is the final location,
		\item $\Clock$ is a finite set of clocks,
		\item $\invariant$ is the invariant, assigning to every $\loc\in \Loc$ a clock guard $\invariant(\loc)$,
		\item $\Edges$ is a finite set of edges  $\edge = (\loc,\guard,\action,\resets,\loc')$
		where~$\loc,\loc'\in \Loc$ are the source and target locations, $\action \in \Actions$, $\resets\subseteq \Clock$ is a set of clocks to be reset, and $\guard$ is a clock guard.
	\end{ienumeration}
\end{definition}
\ifdefined\VersionForArXiV
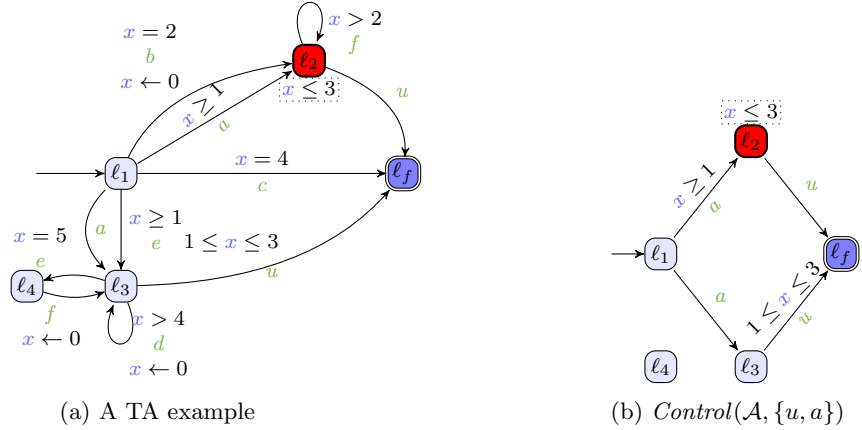
\begin{figure*}[tb]
\else
\begin{figure}[tb]
\fi
\centering
	\begin{subfigure}[b]{0.5\columnwidth}
	\centering
	\footnotesize

		\begin{tikzpicture}[pta, scale=1, xscale=2.5, yscale=1.5]

			\node[location, initial] at (0, 0) (l1) {$\loc_1$};

			\node[location, private] at (1, 1) (l2) {$\loc_2$};

			\node[location, final] at (1.5, 0) (lF) {$\locfinal$};

			\node[location] at (0, -1) (l3) {$\loc_3$};

			\node[location] at (-.5, -1) (l4) {$\loc_4$};

			\node[invariant, below=of l2] {$\clockx \leq 3$};

			\path (l1) edge[] node[sloped,above]{$\clockx \geq 1$} node[below, sloped] {$\styleact{a}$} (l2);
			\path (l1) edge[bend left] node[align=center]{$\clockx = 2$\\$\styleact{b}$\\$\clockx \assign 0$} (l2);
			\path (l1) edge[] node[above]{$\clockx = 4$} node[below]{$\styleact{c}$} (lF);
			\path (l1) edge[] node[align=center]{$\clockx \geq 1$\\$\styleact{e}$} (l3);
			\path (l1) edge[bend right] node[align=center]{$\styleact{a}$} (l3);
			
			\path (l2) edge[loop above] node[below right, xshift=.5em, align=center]{$\clockx > 2$\\$\styleact{f}$} (l2);
			\path (l2) edge[bend left] node[align=center]{$\styleact{u}$} (lF);
			
			\path (l3) edge[bend right] node[above left]{$1 \leq \clockx \leq 3$} node[below left]{$\styleact{u}$} (lF);
			\path (l3) edge[loop below] node[right, align=center]{$\clockx > 4$\\$\styleact{d}$\\$\clockx \assign 0$} (l3);
			\path (l3) edge[bend right] node[above, xshift=-1.5em, align=center]{$\clockx = 5$\\$\styleact{e}$} (l4);
			
			\path (l4) edge[bend right] node[below, xshift=-1em, align=center]{$\styleact{f}$\\$\clockx \assign 0$} (l3);
		\end{tikzpicture}
		\caption{A TA example}
		\label{figure:running-example:TA}
	\end{subfigure}
	\ifdefined\VersionForArXiV
		\hspace{10em}
	\else
	\fi
	\begin{subfigure}[b]{0.45\columnwidth}
	\centering
	\footnotesize

		\begin{tikzpicture}[pta, scale=1, xscale=1.2, yscale=1.5]

			\node[location, initial] at (0, 0) (l1) {$\loc_1$};

			\node[location, private] at (1, 1) (l2) {$\loc_2$};

			\node[location, final] at (2, 0) (lF) {$\locfinal$};

			\node[location] at (1, -1) (l3) {$\loc_3$};

			\node[location] at (0, -1) (l4) {$\loc_4$};

			\node[invariant, above=of l2] {$\clockx \leq 3$};

			\path (l1) edge[] node[sloped,above]{$\clockx \geq 1$} node[below, sloped] {$\styleact{a}$} (l2);
			\path (l1) edge[] node[align=center]{$\styleact{a}$} (l3);
			
			\path (l2) edge[] node[align=center]{$\styleact{u}$} (lF);
			
			\path (l3) edge[] node[sloped, above]{$1 \leq \clockx \leq 3$} node[sloped,below]{$\styleact{u}$} (lF);
		\end{tikzpicture}
		\caption{$\Control(\A, \{ u,a \})$}
		\label{figure:running-example:control}
	\end{subfigure}
	
	\caption{Running example}

\ifdefined\VersionForArXiV
\end{figure*}
\else
\end{figure}
\fi
\begin{example}
	Consider the TA in \cref{figure:running-example:TA}, using one clock~$\clock$.
	$\loc_1$ is the initial location, while we assume that~$\locfinal$ is the  \emph{final} location, \ie{} a location in which an attacker can measure the execution time from the initial location.
	$\loc_2$ is the private location, \ie{} a secret to be preserved: the attacker should not be able to deduce whether it was visited or not.
	$\loc_2$ has an invariant $\clockx \leq 3$ (boxed); other locations invariants are true.
\end{example}
\begin{definition}[Semantics of a TA~\cite{AD94}]\label{def:semantics:TA}
	Given a TA $\TA = (\Actions, \Loc, \locinit, \locpriv, \locfinal, \Clock, \invariant, \Edges)$,
	the semantics of $\TA$ is given by the timed transition system (TTS) %
		$\TTS_{\TA}=(\States, \sinit, \flecheRel)$, with
	\begin{itemize}
		\item $\States = \{ (\loc, \clockval) \in \Loc \times \setRgeqzero^\ClockCard \mid \clockval \models \invariant(\loc) \}$, %
		$\sinit = (\locinit, \ClocksZero) $,
		\item  $\flecheRel$ consists of the discrete and (continuous) delay transition relations:
		\begin{ienumeration}
			\item discrete transitions: $(\loc,\clockval) \longueflecheRel{\edge} (\loc',\clockval')$, %
				if $(\loc, \clockval) , (\loc',\clockval') \in \States$, and there exists ${\edge = (\loc,\guard,\action,\resets,\loc') \in \Edges}$, such that $\clockval'= \reset{\clockval}{\resets}\models \invariant(\loc')$, and $\clockval\models\guard$.
			\item delay transitions: $(\loc,\clockval) \longueflecheRel{d} (\loc, \clockval+d)$, with $d \in \setRgeqzero$, if $\forall d' \in [0, d], (\loc, \clockval+d') \in \States$.
		\end{ienumeration}
	\end{itemize}
\end{definition}

    Moreover we write $(\loc, \clockval)\longuefleche{(d, \edge)} (\loc',\clockval')$ for a combination of a delay and discrete transition if
		$\exists  \clockval'' :  (\loc,\clockval) \longueflecheRel{d} (\loc,\clockval'') \longueflecheRel{\edge} (\loc',\clockval')$.
	\LongVersion{%

}Given a TA~$\TA$ with semantics $(\States, \sinit, \flecheRel)$,
a \emph{run} of~$\TA$ is an alternating sequence of states of $\TTS_{\TA}$  and pairs of delays and edges starting from the initial state~$\sinit$ of the form
$\concstate_0, (d_0, \edge_0), \concstate_1, \cdots$
where
for all $i$, $\edge_i \in \Edges$, $d_i \in \setRgeqzero$ and
	$\concstate_i \longuefleche{(d_i, \edge_i)} \concstate_{i+1}$.
The \emph{duration} of a finite run $\varrun : \concstate_0, (d_0, \edge_0), \concstate_1, \cdots, (d_{i-1}, \edge_{i-1}), (\loc_i, \clockval_i)$ is $\duration(\varrun) = \sum_{0 \leq j \leq i-1} d_j$.
\subsection{Timed opacity definitions}\label{ss:problems}

We recall here the notion of timed opacity defined in~\cite{ALMS23}.\footnote{%
	We slightly modify the definitions from~\cite{ALMS23} by incorporating $\locpriv$ within the definition of~$\A$, and removing $\locpriv$ and~$\locfinal$ from the definition of~$\PrivDurReach{\A}$ and~$\PubDurReach{\A}$ to simplify the reading.
}

Given $\TA = (\Actions, \Loc, \locinit, \locpriv, \locfinal, \Clock, \invariant, \Edges)$,
and a run~$\varrun$, we say that $\locpriv$ is reached on the way to~$\locfinal$ in~$\varrun$ if $\varrun$ is of the form
$(\loc_0, \clockval_0), (d_0, \edge_0), (\loc_1, \clockval_1), \cdots, \ifdefined\VersionForArXiV\linebreak\fi(\loc_m, \clockval_m), (d_m, \edge_m), \cdots (\loc_n, \clockval_n)$ %
	for some~$m,n \in \setN$ such that $\loc_m = \locpriv$, $\loc_n = \locfinal$ and $\forall 0 \leq i \leq m-1, \loc_i \neq \locfinal$.
We denote by $\PrivReach{\TA}{\locpriv}{\locfinal}$ the set of those runs, and refer to them as \emph{private} runs.
Conversely, we say that
$\locpriv$ is avoided on the way to~$\locfinal$ in~$\varrun$ 
if $\varrun$ is of the form $(\loc_0, \clockval_0), (d_0, \edge_0), (\loc_1, \clockval_1), \cdots, (\loc_n, \clockval_n )$ %
with $\loc_n = \locfinal$ and $\forall 0 \leq i < n, \loc_i \not \in \{\locpriv,\locfinal\}$.
We denote the set of those runs by~$\PubReach{\TA}{\locpriv}{\locfinal}$, and refer to them as \emph{public} runs.

While we model the secret behavior of the system using a private location $\locpriv$
here, note that one could easily adapt these definitions if the secret is, for example, a set of 
locations, an action (this will be the case in our case study) or the value of a variable.

$\PrivDurReach{\TA}$ (resp.\ $\PubDurReach{\TA}$) is the set of all the durations of the runs for which $\locpriv$ is reached (resp.\ avoided) on the way to~$\locfinal$.
Formally:
\(\PrivDurReach{\TA} = \{ d \in \setRgeqzero \mid \exists \varrun \in \PrivReach{\TA}{\locpriv}{\locfinal}\text{ such that } d = \duration(\varrun)\}\)
and
\(\PubDurReach{\TA} = \{ d \in \setRgeqzero \mid \exists \varrun \in \PubReach{\TA}{\locpriv}{\locfinal} \text{ such that } d = \duration(\varrun)\}\text{.}\)

\begin{definition}[full timed opacity]\label{definition:opacity}
	Given a TA~$\TA$,
	we say that $\TA$ is \emph{fully timed-opaque}
	if $\PrivDurReach{\TA} = \PubDurReach{\TA}$.
\end{definition}

That is, a system is fully timed-opaque if, for any execution time~$d$, there exists a run of 
duration~$d$ that reaches~$\locfinal$ after going through~$\locpriv$ iff there exists 
another run of duration~$d$ that reaches~$\locfinal$ without going through~$\locpriv$.
Hence, the attacker cannot deduce from the execution time whether~$\locpriv$ was visited or not.

\begin{example}\label{example:running:DurReach}
	Consider again the TA in \cref{figure:running-example:TA}.
	Recall that $\loc_2$ is the private location.
	We have $\PrivDurReach{\TA} = [1, 5]$
	and
	$\PubDurReach{\TA} = [1, 3] \cup [4,4] \cup (5, +\infty) $.
	Since $\PrivDurReach{\TA} \neq \PubDurReach{\TA}$, the system is \emph{not} fully timed-opaque.
\end{example}
\section{Untimed control for full timed opacity}\label{section:control}

In this section, we introduce an untimed control for controlling timed opacity.
We assume $\Actions = \ControllableActions \uplus \UncontrollableActions$ where $\ControllableActions$ (resp.~$\UncontrollableActions$) denote controllable (resp.\ uncontrollable) actions.

A (static, untimed) \emph{strategy} of a TA~$\A$ is a set of actions $\strategy \subseteq \Actions$ that contains at least all uncontrollable actions (\ie{} $\UncontrollableActions \subseteq \strategy \subseteq \Actions$).
A strategy induces a restriction of~$\A$ where only the edges labeled by actions 
of~$\strategy$ are allowed:

\begin{definition}[Controlled TA]\label{def:Control}	
	Given \mbox{$\A = (\Actions, \Loc, \locinit, \locpriv, \locfinal,  \Clock, \invariant, \Edges)$} with $\Actions = {\UncontrollableActions \uplus \ControllableActions}$
	and a strategy~$\strategy \subseteq \Actions$,
	the \emph{control} of~$\A$ using~$\strategy$ is the TA 
	\mbox{$\A' = \Control(\A, \strategy) = (\strategy,  \Loc, \locinit, \locpriv, \locfinal, \Clock, \invariant, \Edges')$} where
$\Edges' = \{ (\loc,\guard,\action,\resets,\loc') \in \Edges \mid \action \in \strategy \}$.
\end{definition}
\begin{example}
	Consider again the TA~$\A$ in \cref{figure:running-example:TA}.
	Fix $\strategy = \{u,a\}$.
	Then $\Control(\A, \strategy)$ is in \cref{figure:running-example:control}. %
\end{example}

Strategies represent some modifications of the system that can be implemented to ensure full timed opacity.

\begin{definition}[fully timed-opaque strategy]
	A strategy $\strategy$ is \emph{fully timed-opaque} if $\Control(\A, \strategy)$ is fully timed-opaque.
\end{definition}

A strategy (even a maximal one) might achieve full timed opacity by blocking all runs (both private or public) from reaching the target.
If reaching the target means completing a task, this might not be something one would desire.
We call a strategy allowing to reach the target for at least some durations an \emph{effective} strategy.

We define two slightly different problems:
taking a TA~$\A$ as input, the \textbf{full timed} (resp.\ \textbf{effective full time}) \textbf{opacity control emptiness problem} asks whether the set of fully (resp.\ effective fully) timed-opaque strategies for~$\A$ is empty.

\LongVersion{
\smallskip

\defProblem
	{Full timed opacity control emptiness}
	{A TA~$\A$}
	{is the set of fully timed-opaque strategies for~$\A$ empty?}

\smallskip

\defProblem
	{Effective full timed opacity control emptiness}
	{A TA~$\A$}
	{is the set of effective fully timed-opaque strategies for~$\A$ empty?}
}

Note that, due to the presence of uncontrollable actions, the first problem (full timed opacity control emptiness) is not trivial.
(If uncontrollable actions were not part of our definitions, choosing $\strategy = \emptyset$ would always yield an acceptable fully timed-opaque strategy.)

We will also refine those problems by considering a notion of \emph{maximal} (\ie{} most permissive) strategy \wrt{} full timed opacity based on the number of actions belonging to the strategy:
given~$\A$, a fully timed-opaque strategy~$\strategy$ is maximal if $\forall \strategy'$, if $\strategy'$ is fully timed-opaque then $|\strategy'| \leq |\strategy|$.
We define similarly \emph{minimal} strategies (least permissive, \ie{} disabling as many actions as possible) as well as maximal (resp.\ minimal) effective fully timed-opaque strategies, \ie{} the set of largest (resp.\ smallest) effective fully timed-opaque strategies.

\begin{example}\label{example:running:strategy}
	Consider again the TA~$\A$ in \cref{figure:running-example:TA}.
	Assume $\UncontrollableActions = \{ u \}$ and $\ControllableActions = \{a,b,c,d,e,f\}$.
	Fix $\strategy_1 = \{ u, b, c \}$.
	We have $\PrivDurReach{\Control(\A, \strategy_1)} = [2,5]$
	while
	$\PubDurReach{\Control(\A, \strategy_1)} = [4,4]$;
	therefore, $\strategy_1$ is not fully timed-opaque. 
	Now fix $\strategy_2 = \{ u, a, f \}$.
	We have $\PrivDurReach{\Control(\A, \strategy_2)} = \PubDurReach{\Control(\A, \strategy_2)} = [1,3]$;
	therefore, $\strategy_2$ is fully timed-opaque.

	In fact, it can be shown that the set of effective fully timed-opaque strategies for~$\A$~ is $\{ \{ u,a \} , \{ u,a,e\} , \{ u,a,f\} \}$; 
	therefore, $\{ u,a \}$ is the only minimal strategy, while $\{ u,a,e\} , \{ u,a,f\} $ are the two maximal strategies.
	In addition, $\{ u,f \}$ is an example of a strategy that is not effective, as $\locfinal$ is always unreachable, whether $\locpriv$ is visited or not.
\end{example}

\LongVersion{
\subsection{Complexity}
}

\begin{proposition}[complexity]\label{prop:complexity}
	One can compute the set of fully timed-opaque strategies over a TA~$\TA$
in \TEXPTIME.
\end{proposition}
\begin{proof}
The full timed opacity decision problem (\ie{} checking if a given TA is fully timed-opaque) is decidable for TAs in (at most) \TEXPTIME~\cite{ALMS23}.
Moreover, reachability of the final state can be decided in \PSPACE~\cite{AD94}.
Thus, for any given strategy, one can check in triple exponential time whether it is
(effective) fully timed-opaque. 

Computing the list of (effective) fully timed-opaque strategies can be done naively by 
testing each possible strategy one by one and keeping the ones that satisfy the property\LongVersion{ we want}.
As there is an exponential number of possible strategies and repeating exponentially many 
times a \TEXPTIME algorithm remains in \TEXPTIME, this algorithm is in 
\TEXPTIME.
\end{proof}

As a corollary of the above, the (effective) full timed opacity control emptiness problem is in \TEXPTIME as well.
More precisely, the above proof establishes that the complexity class of the 
(effective) full timed opacity control emptiness problem is the maximum between 
\PSPACE and the complexity of the full timed opacity problem. As the latter is 
\PSPACE-hard (being trivially harder than reachability), the two problems lie
in the same complexity class. 
From a theoretical point of view, one thus cannot do better than the naive enumeration 
approach described here to solve the control problem.

Finding the maximal (resp. minimal) strategies can be done slightly more efficiently by 
starting from the set with every (resp. no) controllable action and enumerating the
potential strategies by decreasing (resp. increasing) order as one could then 
potentially stop before full enumeration. In the worst case, this will however have the 
same complexity as the full enumeration.

\section{Implementation and experiments}\label{section:implementation}

\LongVersion{
\subsection{Implementation in \strategFTO{}}
}

We implemented our strategy generation in \strategFTO{}, an entirely automated open-source tool written in Java.\footnote{%
	Source code is \LongVersion{available }at \url{https://github.com/DylanMarinho/Controlling-TA}.
	Models and experiment results are available at \href{https://doi.org/10.5281/zenodo.7181848}{\nolinkurl{10.5281/zenodo.7181848}}.
}
Our tool iteratively constructs strategies, then checks full timed opacity\LongVersion{ following \cref{algo:synth-control}}\ACMVersion{ by computing the private and public execution times and by checking their equality}.

\newcommand{\displayAlgoSynthControl}{%
\begin{algorithm}[tb]

	$\strategies \assign \emptyset$
	
	\ForEach{$ s \subseteq \ControllableActions $}{

		$\strategy \assign s \cup \UncontrollableActions$
		
		\tcc{Compute execution times}
		$\lambda_1 \assign \PubDurReach{\Control(\A, \strategy)}$\label{algo:synth-control:PubDurReach}

		$\lambda_2 \assign \PrivDurReach{\Control(\A, \strategy)}$\label{algo:synth-control:PrivDurReach}

		\tcc{Check for full timed opacity}
		\lIf{$\lambda_1 = \lambda_2$}{
			$\strategies \assign \strategies \cup \set{\strategy}$
		}		
	}
	
	\Return{\strategies}
	
	\caption{$\synthControl(\A)$ Exhibit all timed-opaque strategies}
	\label{algo:synth-control}
\end{algorithm}
}
\LongVersion{
	\displayAlgoSynthControl{}
}

\LongVersion{We give our strategy synthesis algorithm in \cref{algo:synth-control}.}
The exhibition of these execution times ($\PubDurReach{\A}$ and $\PrivDurReach{\A}$\LongVersion{, \cref{algo:synth-control:PubDurReach}}) is done in our implementation by an automated model modification (following the procedure described in~\cite{ALMS23}, but which was not entirely automated in~\cite{ALMS23}) followed by a synthesis problem using a parametric extension of TAs~\cite{AHV93}.
The synthesis of the execution times itself is done by a call to an external tool---\imitator{}~3.3 \emph{``Cheese Caramel au beurre salé''}~\cite{Andre21}.%
\LongVersion{%

}%
\ACMVersion{ }
\strategFTO{} then checks whether both sets of execution times are equal; this is done by a call to another external tool---\PolyOp{}~1.2\footnote{\url{https://github.com/etienneandre/PolyOp}}, that performs polyhedral operations\LongVersion{ as a simple interface for the Parma polyhedra library}\ACMVersion{ using PPL}~\cite{BHZ08}.

\paragraph{Algorithms}
We implement not only the exhibition of all timed-opaque strategies (denoted by $\synthControl(\A)$\LongVersion{, in \cref{algo:synth-control}}), but also the following variants:
\begin{ienumeration}%
	\item $\synthMaxControl(\A)$: synthesize all maximal strategies for~$\A$;
	\item $\synthMinControl(\A)$: synthesize all minimal strategies;
	\item $\witnessMaxControl(\A)$: witness \emph{one} maximal strategy;
	\item $\witnessMinControl(\A)$: witness \emph{one} minimal strategy.
\end{ienumeration}%
We implemented these other algorithms by changing the exploration order of the strategies, and/or by triggering immediate termination upon the first exhibition of a strategy.

\paragraph{Input model}
The input TA model is given in the \imitator{} input syntax;
while we presented a restricted setting in this paper for sake of clarity, our implementation in \strategFTO{} is much more permissive, by allowing significant extensions of TAs with global (integer or Boolean) variables, multiple automata with synchronization, multi-rate clocks (including stopwatches), etc.

\subsection{Proof of concept benchmark}
\begin{figure*}[tb]
	\centering
	\scriptsize
	
	\scalebox{\ifdefined\VersionForArXiV .95\else .9\fi}{
		\begin{tikzpicture}[pta, scale=1, xscale=2.5, yscale=1.3]
			
			\node[location, initial] at (-2, -2) (initial) {$I$};
			
			\node[location] at (-1, -2) (welcome) {$W$};
			\node[invariant, above=of welcome] {$\clockx \leq 3$};
			
			\node[location] at (0, -2) (waitingPassword) {$WP$};
			\node[invariant, right=of waitingPassword] {$\clockx \leq 10$};
			
			\node[location] at (0, -3) (waitChoice) {$WC$};
			\node[invariant, right=of waitChoice] {$\clockx \leq 10$};
			
			\node[location] at (0, -4) (waitingAmount) {$WA$};
			\node[invariant, right=of waitingAmount, yshift=.5em] {$\clockx \leq 10$};
			
			\node[location] at (0, -5) (preparingWithdrawalNormal) {$PNW$};
			\node[invariant, right=of preparingWithdrawalNormal] {$\clockx \leq 15$};
			
			\node[location] at (-2, -5) (preparingWithdrawalQuick) {$PQW$};
			\node[invariant, right=of preparingWithdrawalQuick] {$\clockx \leq 15$};
			
			\node[location] at (-3.5, -5) (displayingBalance) {$DB$};
			\node[invariant, right=of displayingBalance] {$\clockx \leq 10$};
			
			\node[location] at (0, -6) (moneyAvailableNormal) {$MAN$}; %
			\node[invariant, right=of moneyAvailableNormal] {$\clockx \leq 20$};
			
			\node[location] at (-2, -6) (moneyAvailableQuick) {$MAQ$}; %
			\node[invariant, right=of moneyAvailableQuick] {$\clockx \leq 20$};
			
			\node[location] at (-1, -7) (otherOperation) {$OO$};
			\node[invariant, below=of otherOperation] {$\clockx \leq 10$};
			
			\node[location] at (0, -7.5) (terminating) {$T$};
			\node[invariant, right=of terminating] {$\clocky \leq 100$};
			
			\node[location] at (2, -8) (cancelling) {$C$};
			\node[invariant, right=of cancelling] {$\clocky \leq 100$};
			
			\node[location, final] at (-1, -8) (theEnd) {$E$};

			\path (initial) edge[] node[above]{$\ATMstart$} node[below]{$\clockx, \clocky \assign 0$} (welcome);
			
			\path (welcome) edge[] node[above, align=center]{$\clockx = 3$\\$\ATMaskPassword$} node[below]{$\clockx \assign 0$} (waitingPassword);
			
			\path (waitingPassword) edge[loop above] node[right, align=center]{$\styledisc{nbFP} < 3$\\$\ATMincorrectPassword$\\$\styledisc{nbFP}++$} (waitingPassword);
			
			\path (waitingPassword) edge[] node[left, align=center]{$\ATMcorrectPassword$\\$\clockx \assign 0$} (waitChoice);
			
			\path (waitingPassword) edge[bend left] node[sloped, align=center]{$\clockx = 10$} (cancelling);
			
			\path (waitChoice) edge[bend right] node[left, align=center]{$\ATMquickWithdrawal$\\$\clockx \assign 0$} (preparingWithdrawalQuick);
			
			\path (waitChoice) edge[] node[align=center]{$\ATMnormalWithdrawal$\\$\clockx \assign 0$} (waitingAmount);
			
			\path (waitChoice) edge[bend right] node[above left, align=center]{$\ATMrequestBalance$\\$\clockx \assign 0$} (displayingBalance);
			
			\path (waitingAmount) edge[loop left] node[align=center]{$\styledisc{nbFA} < 3$\\$\ATMincorrectAmount$\\$\styledisc{nbFA}++$} (waitingAmount);
			
			\path (waitingAmount) edge[] node[align=center]{$\ATMcorrectAmount$\\$\clockx \assign 0$} (preparingWithdrawalNormal);
			
			\path (waitingAmount) edge[out=0, in=90] node[sloped, align=center]{$\clockx = 12$} (cancelling);
			
			\path (preparingWithdrawalNormal) edge[] node[align=center]{$\clockx = 15$\\$\clockx \assign 0$} (moneyAvailableNormal);
			
			\path (preparingWithdrawalQuick) edge[] node[align=center]{$\clockx = 15$\\$\clockx \assign 0$} (moneyAvailableQuick);
			
			\path (moneyAvailableNormal) edge[] node[sloped, align=center]{$\clockx = 20$} (cancelling);
			
			\path (moneyAvailableNormal) edge[] node[sloped, above, xshift=1em]{$\ATMtakeCash$} node[sloped, below, xshift=2em,yshift=-1em]{$\clockx \assign 0$} (otherOperation);
			
			\path (moneyAvailableQuick) edge[] node[above, sloped, align=center]{$\clockx = 20$} (cancelling);
			
			\path (moneyAvailableQuick) edge[] node[sloped, below, align=center]{$\ATMtakeCash$} (theEnd);
			
			\path (displayingBalance) edge[bend right] node[sloped, above]{$\clockx = 10$} node[sloped, below]{$\clockx \assign 0$} (otherOperation);
			
			\path (displayingBalance) edge[bend right] node[below left, align=center]{$\ATMpressOK$} (theEnd);
			
			\path (otherOperation) edge[bend left] node[sloped,align=center]{$\clockx = 10$} (terminating);
			\path (otherOperation) edge[bend right] node[sloped, %
			]{$\ATMpressFinish$} (terminating);
			
			\path (otherOperation) edge[] node[sloped, align=center]{$\ATMrestart$} (waitChoice);
			
			\path (cancelling) edge[] node[align=center]{$\clocky = 100$\\$\ATMfinish$} (theEnd);
			
			\path (terminating) edge[] node[align=center]{$\clocky = 100$\\$\ATMfinish$} (theEnd);

			\node[] at (2.6, -2.6) (legend) {
				\begin{tabular}{l @{~:~} l}
					$C$ & cancelling\\
					$DB$ & displaying balance\\
					$E$ & end\\
					$I$ & initial\\
					$MAN$ & money available normal\\
					$MAQ$ & money available quick\\
					$OO$ & other operations\\
					$PNW$ & preparing normal withdrawal\\
					$PQW$ & preparing quick withdrawal\\
					$T$ & terminating\\
					$W$ & waiting\\
					$WA$ & waiting for amount\\
					$WC$ & waiting for choice\\
					$WP$ & waiting for password\\
				\end{tabular}
			};

		\end{tikzpicture}
	}
	
	\caption{ATM benchmark}
	\label{figure:ATM}
\end{figure*}
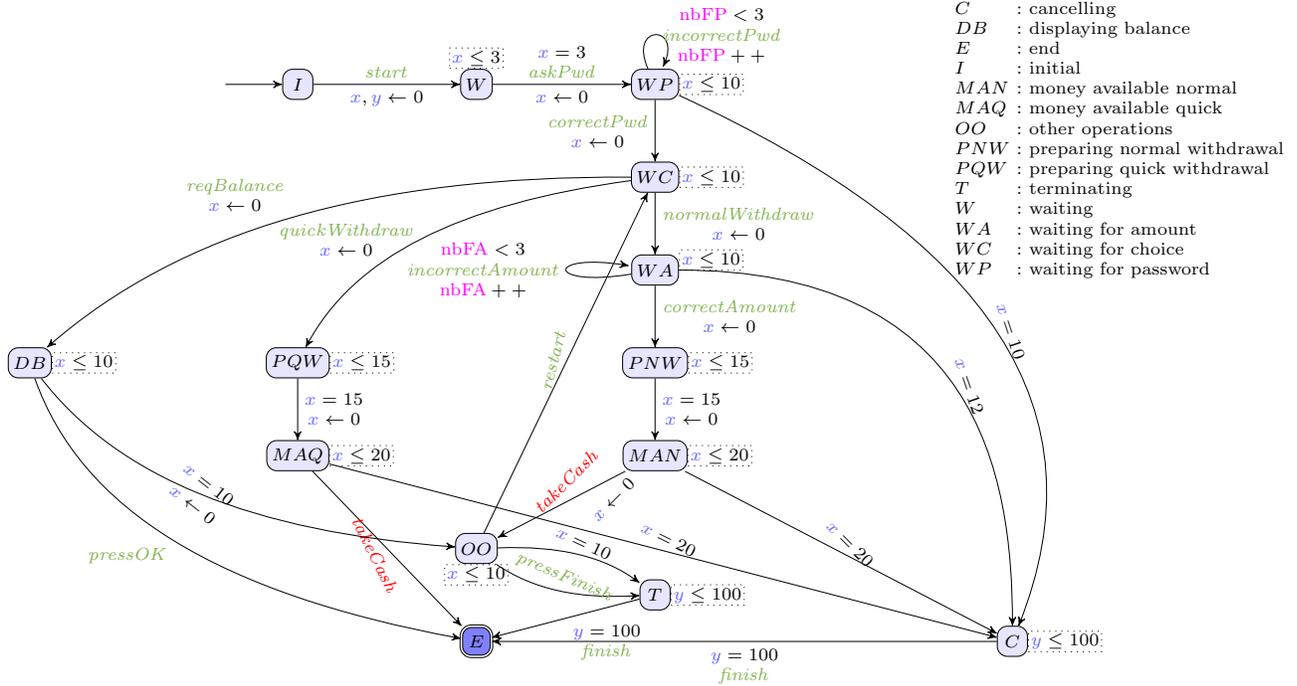
As a proof of concept, we consider the TA model of an ATM (given in \cref{figure:ATM}).
The idea is that (as per our definition of timed opacity) the attacker only has access to the execution time, \ie{} the time from the beginning of the program to reaching the end state.
The secret is whether the ATM user has actually obtained cash (action $\ATMtakeCash$).\footnote{%
	\strategFTO{} allows not only private locations, but also actions.
}
The TA uses two clocks: $\clockx$ for ``local'' actions, and $\clocky$ for a global time measurement.
First, the user starts the process (action \ATMstart), then the ATM displays a welcome screen for 3~time units, followed by another screen requesting the password (action \ATMaskPassword).
Then, the user can submit a correct (action \ATMcorrectPassword{}) or incorrect (\ATMincorrectPassword{}) password; if no password is input within 10 time units, the system moves to a cancelling phase.
The same happens if 3 incorrect password have been input.
After inputting the correct password, the user has the choice between a fixed-amount quick withdrawal (\ATMquickWithdrawal), a normal withdrawal (\ATMnormalWithdrawal) or a balance request (\ATMrequestBalance).

The quick withdrawal triggers a 15-time unit preparation followed by the availability of the money, which the user can take immediately (action \ATMtakeCash), thus terminating the procedure.
If the user does not take the money, the system moves to the cancelling phase.

The normal withdrawal asks the user to input the desired amount; similar to the password, after 3 wrong amounts (action \ATMincorrectAmount), or upon timeout, the system moves to cancelling phase.
After the user retrieves cash (action \ATMtakeCash), they are asked whether they would like to perform another operation; if so (action \ATMrestart), the system goes back to the choice location.
Otherwise (action \ATMpressFinish), or unless a 10-time unit timeout is reached, the system moves to the terminating location.
The balance request triggers the balance display, from which the user can immediately terminate the process (action \ATMpressOK), or go back to the choice menu.

The rationale is that, in the regular terminating and cancelling phases, the ATM terminates after constant time (invariant $\clocky \leq 100$), avoiding leaking information.
However, some actions may lead to quicker termination (quick withdrawal) or slower termination (multiple choices).

The uncontrollable actions are most of the user actions: \ATMcorrectAmount{}, \ATMincorrectAmount{}, \ATMcorrectPassword, \ATMincorrectPassword{}, \ATMpressFinish{}, \ATMtakeCash{}.
The controllable actions are the system actions (\ATMaskPassword{}, \ATMstart{}, \ATMfinish{}) and some of the users actions that can be controlled by disabling the associated choice (\ATMrequestBalance{}, \ATMpressOK{}, \ATMquickWithdrawal{}, \ATMrestart{}).

\subsection{Experiments}

We \LongVersion{first }exhibit in \cref{table:strategies} controllers for \LongVersion{our benchmark from \cref{figure:ATM}}\ACMVersion{our \cref{figure:ATM} benchmark} \LongVersion{as }computed by \strategFTO{}, for all \LongVersion{our }algorithms.
For space concern, we tabulate the actions to \emph{disable}; the strategy is therefore $\Actions$ minus these actions.
Also note that, for $\witnessMaxControl$ and $\witnessMinControl$, the \emph{order} in which we compute the subsets of~$\Actions$\LongVersion{ in \cref{algo:synth-control}} has an impact on the result, as the algorithm stops as soon as \emph{one} strategy is found.
According to \cref{table:strategies}, the maximal strategies (\ie{} the most permissive, disabling the least number of actions) are to disable either $\ATMrestart$ and~$\ATMpressOK$, or $\ATMrestart$ and~$\ATMrequestBalance$.
This is natural, as $\ATMrestart$ allows the user to restart a second operation, thus violating the constant-time nature of \cref{figure:ATM}, while $\ATMpressOK$ and $\ATMrequestBalance$, if enabled together, allow a quick exit, shorter than a cash withdrawal operation---thus giving hint to the attacker that the $\ATMtakeCash$ secret did \emph{not} occur.

\begin{table*}[tb]
\caption{Strategy synthesis for \cref{figure:ATM}}
\centering
\setlength{\tabcolsep}{1pt}%
\scriptsize%
\begin{tabular}{@{} l l l l l l}
	\rowHeader{} Actions to disable                               & $\synthMinControl$ & $\witnessMinControl$ & $\synthMaxControl$ & $\witnessMaxControl$ & $\synthControl$ \\
	\ATMrestart, \ATMpressOK                                      &                    &                             & \cellYes           & \cellYes                    & \cellYes           \\
	\ATMrestart, \ATMrequestBalance                                   &                    &                             & \cellYes           &                             & \cellYes           \\
	\ATMrestart, \ATMpressOK, \ATMquickWithdrawal                 &                    &                             &                    &                             & \cellYes           \\
	\ATMrestart, \ATMpressOK, \ATMrequestBalance                      &                    &                             &                    &                             & \cellYes           \\
	\ATMrestart, \ATMquickWithdrawal, \ATMrequestBalance              &                    &                             &                    &                             & \cellYes           \\
	\ATMrestart, \ATMpressOK, \ATMquickWithdrawal, \ATMrequestBalance & \cellYes           & \cellYes                    &                    &                             & \cellYes
\end{tabular}
\label{table:strategies}
\end{table*}
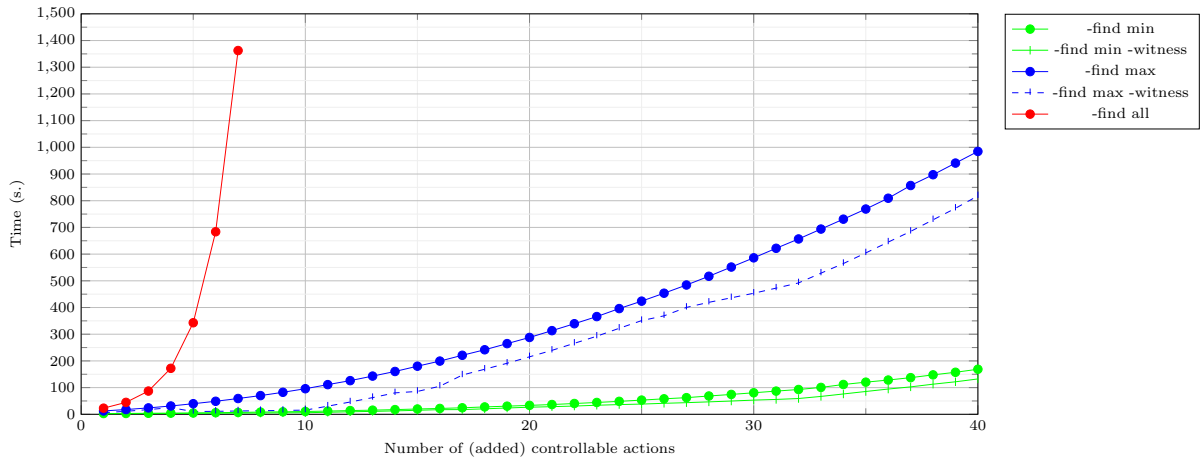
\begin{figure*}[tb]
	\centering
	\scriptsize
	\begin{tikzpicture}[scale=.8]
		\begin{axis}[
			xmin = 0, xmax = 40,
			ymin = 0, ymax = 1500,
			xtick distance = 10,
			ytick distance = 100,
			xlabel=Number of (added) controllable actions,
			ylabel=Time (s.),
			grid = both,
			minor tick num = 1,
			major grid style = {lightgray},
			minor grid style = {lightgray!25},
			width = \textwidth,
			height = 0.5\textwidth,
			legend pos = outer north east
			]

			\addplot[
				legend entry=-find min,
				color=green, mark=*,
			] table [x=Actions, y=-find min, col sep=comma] {results/result-graph-ATM-scalable-50a-1800TO.csv};

			\addplot[
			legend entry=-find min -witness,
			color=green,mark=|,
			] table [x=Actions, y=-find min -witness, col sep=comma] {results/result-graph-ATM-scalable-50a-1800TO.csv};

			\addplot[
			legend entry=-find max,
			color=blue, mark=*,
			] table [x=Actions, y=-find max, col sep=comma] {results/result-graph-ATM-scalable-50a-1800TO.csv};

			\addplot[
			legend entry=-find max -witness,
			color=blue,mark=|,dashed,
			] table [x=Actions, y=-find max -witness, col sep=comma] {results/result-graph-ATM-scalable-50a-1800TO.csv};

			\addplot[
			legend entry=-find all,
			color=red,mark=*,
			] table [x=Actions, y=-find all, col sep=comma] {results/result-graph-ATM-scalable-50a-1800TO.csv};
		\end{axis}
	\end{tikzpicture}

	\caption{Execution times for scalability (in seconds; TO set at 1,800\,s)}
	\label{fig:plot-scalability}
\end{figure*}

\paragraph{Scalability}
Then, we test the scalability of \strategFTO{} \wrt{} the number of actions.
We modify \cref{figure:ATM} by adding an increasingly large numbers of controllable actions; these actions do not play a role in the control (we basically add unguarded self-loops) but they will impact the computation time, as we will need to consider an increasingly (and exponentially) larger number of subsets of actions\LongVersion{, from \cref{algo:synth-control}}.
We add from 1 to~40 such actions, resulting (by adding the actions in \cref{figure:ATM}) in a model with a number of controllable actions from 11 to~50.
We plot these results in \cref{fig:plot-scalability}.
\LongVersion{(Raw results are in \cref{table:result-scalability} in \cref{appendix:scalability}.)}
From our results in \cref{fig:plot-scalability}, we see that, without surprise, the execution time for $\synthControl$ is exponential in the number of actions.
However, $\synthMaxControl$ and $\synthMinControl$ behave much better, by remaining respectively below 15~minutes and three~minutes, even for up to~50 controllable actions.
In addition, it is important to notice that $\witnessMaxControl$ and $\witnessMinControl$ do not decrease the time very much compared to the full versions $\synthMaxControl$ and $\synthMinControl$. This is because, at a given size, the number of strategies to be tested remains relatively small.

\section{Conclusion}\label{section:conclusion}

We introduced a prototype tool \strategFTO{} implementing an algorithm to exhibit strategies to guarantee the full timed opacity of a system modeled by a timed automaton where the attacker only has access to the computation time.
Even though relying on a simple enumeration of the subsets, our tool \strategFTO{} shows good performance for synthesizing maximal or minimal strategies, with very reasonable times, even for several dozens of controllable actions.

\paragraph{Future works}
We plan to further optimize our implementation by maintaining a set of non-effective strategies, \ie{} for which $\locfinal$ is unreachable:
any strategy strictly included into a known non-effective strategy will necessarily be non-effective too, and therefore no full timed-opacity analysis is needed for this strategy.
An option to efficiently represent this strategies set could be to store it using~BDDs.
	We also plan to strengthen strategies so that their choice may depend on how long has passed since the start of the execution.
	As these strategies still need a finite representation to be handled, this requires establishing exactly what strategies need to remember to chose optimally.

Our ultimate goal will be to extend timed automata to \emph{parametric} timed automata~\cite{AHV93}, and use automated parameter synthesis techniques (\eg{} \cite{JLR15,AAPP21,AMP21}), with a parametric timed controller~\cite{JLR19,Gol21}.

\ifdefined\VersionForArXiV\else
	\begin{acks}
		\ouracks{}
	\end{acks}
\fi

\newcommand{\CCIS}{Communications in Computer and Information Science}
\newcommand{\ENTCS}{Electronic Notes in Theoretical Computer Science}
\newcommand{\FI}{Fundamenta Informormaticae}
\newcommand{\FMSD}{Formal Methods in System Design}
\newcommand{\IJFCS}{International Journal of Foundations of Computer Science}
\newcommand{\IJSSE}{International Journal of Secure Software Engineering}
\newcommand{\IPL}{Information Processing Letters}
\newcommand{\JLAP}{Journal of Logic and Algebraic Programming}
\newcommand{\JLC}{Journal of Logic and Computation}
\newcommand{\LMCS}{Logical Methods in Computer Science}
\newcommand{\LNCS}{Lecture Notes in Computer Science}
\newcommand{\RESS}{Reliability Engineering \& System Safety}
\newcommand{\STTT}{International Journal on Software Tools for Technology Transfer}
\newcommand{\TOSEM}{ACM Transactions on Software Engineering and Methodology}
\newcommand{\TCS}{Theoretical Computer Science}
\newcommand{\ToPNoC}{Transactions on Petri Nets and Other Models of Concurrency}
\newcommand{\TSE}{IEEE Transactions on Software Engineering}

\ifdefined\VersionForArXiV

	\renewcommand*{\bibfont}{\footnotesize}
	\printbibliography[title={References}]
\else
	\bibliographystyle{ACM-Reference-Format}
	\bibliography{PTA.bib}
\fi

\ifdefined\VersionForArXiV

\newpage
\appendix

\begin{center}
	\bfseries\huge
	Appendix
\end{center}

\tikzset{
	property/.style={
		fill=orange!30,
		rounded corners
	},
	workTA/.style={
		fill=green!20,
		rounded corners
	},
	imitator/.style={
	draw,circle
	},
	modelToIMITATOR/.style={
	},
	pubPropToIMITATOR/.style={
		orange
	},
	privPropToIMITATOR/.style={
		orange
	},
	imitatorRes/.style={
		draw,fill=gray,circle
	},
	imitatorExec/.style={
		dashed
	},
	polyop/.style={
		draw,circle
	},
	resToPolyop/.style={
	},
	polyopRes/.style={
		draw,fill=gray,circle
	},
	polyopExec/.style={
		dashed
	},
	answer/.style={
		draw
	},
	extractStrat/.style={
	},
}

\section{Experiments: scalability test}\label{appendix:scalability}

We performed a sample scalability test on our benchmark.
The plot is given in \cref{fig:plot-scalability}.

We tabulate our full results in \cref{table:result-scalability}.

\begin{table*}[!hbt]
	\caption{Execution times for scalability (in seconds; TO set at 1,800\,s)}
	\scriptsize
	\centering
	\resizebox{.7\textwidth}{!}{
	\setlength{\tabcolsep}{2pt} %
	\begin{tabular}{|l|l|l|l|l|l|}
		\rowHeader{} Number of & $\synthMinControl$ & $\witnessMinControl$        & $\synthMaxControl$ & $\witnessMaxControl$        & $\synthControl$    \\
		\rowHeader{} added actions & \texttt{-find min} & \texttt{-find min -witness} & \texttt{-find max} & \texttt{-find max -witness} & \texttt{-find all} \\ \hline
		1 & 2.89 & 2.04 & 12.61 & 5.92 & 22.98 \\ \hline
		2 & 3.19 & 2.44 & 17.81 & 11.30 & 44.68 \\ \hline
		3 & 3.84 & 2.74 & 23.99 & 17.58 & 87.07 \\ \hline
		4 & 4.43 & 2.85 & 31.15 & 24.92 & 172.26 \\ \hline
		5 & 4.90 & 3.77 & 39.69 & 10.34 & 342.95 \\ \hline
		6 & 6.07 & 4.09 & 48.78 & 11.12 & 683.77 \\ \hline
		7 & 7.02 & 4.54 & 59.14 & 12.35 & 1,362.48 \\ \hline
		8 & 8.34 & 4.69 & 70.09 & 13.46 & TO \\ \hline
		9 & 9.32 & 5.63 & 82.45 & 14.52 & TO \\ \hline
		10 & 10.51 & 5.91 & 95.86 & 15.65 & TO \\ \hline
		11 & 12.04 & 7.66 & 111.16 & 30.49 & TO \\ \hline
		12 & 13.99 & 9.43 & 126.11 & 46.00 & TO \\ \hline
		13 & 15.54 & 10.94 & 143.15 & 62.50 & TO \\ \hline
		14 & 17.58 & 12.83 & 160.41 & 80.32 & TO \\ \hline
		15 & 19.88 & 15.03 & 180.24 & 85.64 & TO \\ \hline
		16 & 21.94 & 17.38 & 199.34 & 105.15 & TO \\ \hline
		17 & 24.39 & 17.63 & 221.04 & 146.98 & TO \\ \hline
		18 & 27.64 & 20.53 & 241.70 & 168.79 & TO \\ \hline
		19 & 30.49 & 23.65 & 264.72 & 191.16 & TO \\ \hline
		20 & 33.43 & 26.59 & 287.85 & 215.33 & TO \\ \hline
		21 & 36.58 & 28.60 & 313.34 & 239.52 & TO \\ \hline
		22 & 40.46 & 30.92 & 339.24 & 265.90 & TO \\ \hline
		23 & 44.31 & 33.92 & 366.07 & 292.51 & TO \\ \hline
		24 & 48.52 & 36.43 & 395.95 & 322.16 & TO \\ \hline
		25 & 53.21 & 38.30 & 423.86 & 350.92 & TO \\ \hline
		26 & 58.02 & 41.21 & 453.59 & 368.86 & TO \\ \hline
		27 & 62.36 & 43.57 & 484.28 & 400.39 & TO \\ \hline
		28 & 68.56 & 46.60 & 517.03 & 419.43 & TO \\ \hline
		29 & 74.39 & 49.74 & 551.58 & 436.35 & TO \\ \hline
		30 & 80.64 & 53.15 & 586.03 & 453.37 & TO \\ \hline
		31 & 86.89 & 55.72 & 621.83 & 472.63 & TO \\ \hline
		32 & 92.91 & 59.14 & 656.75 & 492.10 & TO \\ \hline
		33 & 100.67 & 67.00 & 693.82 & 528.71 & TO \\ \hline
		34 & 111.66 & 76.17 & 730.93 & 564.68 & TO \\ \hline
		35 & 120.37 & 85.48 & 768.87 & 604.42 & TO \\ \hline
		36 & 128.35 & 94.59 & 809.36 & 645.01 & TO \\ \hline
		37 & 137.28 & 102.77 & 856.98 & 685.03 & TO \\ \hline
		38 & 147.68 & 112.83 & 897.39 & 728.48 & TO \\ \hline
		39 & 157.42 & 121.77 & 940.98 & 771.68 & TO \\ \hline
		40 & 168.74 & 132.45 & 984.65 & 818.25 & TO \\ \hline
	\end{tabular}
	}
	\label{table:result-scalability}
\end{table*}

\fi

\end{document}